\documentclass[11pt,a4paper]{article}

\usepackage{fullpage}
\usepackage[utf8]{inputenc}
\usepackage{amssymb}
\usepackage[leqno]{amsmath}
\makeatletter
\newcommand{\leqnomode}{\tagsleft@true\let\veqno\@@leqno}
\newcommand{\reqnomode}{\tagsleft@false\let\veqno\@@eqno}
\makeatother
\usepackage{mathabx}
\usepackage{mathtools}
\usepackage{csquotes}
\usepackage[all]{nowidow}
\usepackage{amsthm,amsmath,amssymb}
\usepackage{bm}
\usepackage[mathcal,mathscr]{eucal}
\usepackage{epsfig,graphicx,graphics,color}
\usepackage{tcolorbox}
\usepackage{enumerate}
\usepackage[sort,nocompress]{cite}
\usepackage{hyperref}
\hypersetup{
    colorlinks=true,       
    linkcolor=blue,        
    citecolor=red,         
    filecolor=magenta,     
    urlcolor=cyan,         
    linktocpage=true
}

\usepackage{algorithm}
\usepackage{bbm}
\usepackage[noend]{algpseudocode}
\usepackage{varwidth}
\newcounter{algsubstate}


\theoremstyle{plain}
\newtheorem{theorem}{Theorem}

\newtheorem{corollary}[theorem]{Corollary}

\theoremstyle{definition}

\newtheorem{remark}[theorem]{Remark}

\newcommand{\cC}{\mathcal{C}}

\newcommand{\cX}{\mathcal{X}}

\def\final{0}  
\ifnum\final=0  
\newcommand{\krnote}[1]{{\color{red}[{\tiny Krist\'of: \bf #1}]\marginpar{\color{red}*}}}
\newcommand{\enote}[1]{{\color{blue}[{\tiny Endre: \bf #1}]\marginpar{\color{blue}*}}}
\newcommand{\knote}[1]{{\color{orange}[{\tiny Kaz: \bf #1}]\marginpar{\color{orange}*}}}
\else 
\newcommand{\krnote}[1]{}
\newcommand{\enote}[1]{}
\newcommand{\knote}[1]{}
\fi

\title{Generating clause sequences of a CNF formula}

\author{Krist{\'o}f B{\'e}rczi\thanks{MTA-ELTE Egerv\'ary Research Group, Department of Operations Research, E{\"o}tv{\"o}s Lor{\'a}nd University, Budapest, Hungary. Email: \texttt{berkri@cs.elte.hu}.} 
\and
Endre Boros\thanks{MSIS Department and RUTCOR, Rutgers University, New Jersey, USA. Email: \texttt{endre.boros@rutgers.edu}.}
\and
Ond\v{r}ej \v{C}epek\thanks{Charles University, Faculty of Mathematics and Physics, Department of Theoretical Computer Science and Mathematical Logic, Praha, Czech Republic. Email: \texttt{cepek@ktiml.mff.cuni.cz}.}
\and
Khaled Elbassioni\thanks{Masdar Institute, Khalifa University of Science and Technology, P.O. Box 54224, Abu Dhabi, UAE. Email:
\texttt{khaled.elbassioni@ku.ac.ae}}
\and
Petr Ku\v{c}era\thanks{Charles University, Faculty of Mathematics and Physics, Department of Theoretical Computer Science and Mathematical Logic, Praha, Czech Republic. Email: \texttt{kucerap@ktiml.mff.cuni.cz}.}
\and
Kazuhisa Makino\thanks{Research Institute for Mathematical Sciences (RIMS) Kyoto University, Kyoto, Japan. Email: \texttt{makino@kurims.kyoto.ac.jp}.}
}

\begin{document}
\date{}
\maketitle

\begin{abstract}
Given a CNF formula $\Phi$ with clauses $C_1,\ldots,C_m$ and variables $V=\{x_1,\ldots,x_n\}$, a truth assignment $\bm{a}:V\rightarrow\{0,1\}$ of $\Phi$ leads
to a clause sequence $\sigma_\Phi(\bm{a})=(C_1(\bm{a}),\ldots,C_m(\bm{a}))\in\{0,1\}^m$ where $C_i(\bm{a}) = 1$ if clause $C_i$ evaluates to $1$ under assignment $\bm{a}$, otherwise $C_i(\bm{a}) = 0$. The set of all possible clause sequences carries a lot of information on the formula, e.g. SAT, MAX-SAT and MIN-SAT can be encoded in terms of finding a clause sequence with extremal properties.

We consider a problem posed at Dagstuhl Seminar 19211 ``Enumeration in Data Management'' (2019) about the generation of all possible clause sequences of a given CNF with bounded dimension. We prove that the problem can be solved in incremental polynomial time. We further give an algorithm with polynomial delay for the class of tractable CNF formulas. We also consider the generation of maximal and minimal clause sequences, and show that generating maximal clause sequences is NP-hard, while minimal clause sequences can be generated with polynomial delay.
  
  \bigskip

\end{abstract}

\section{Introduction}
\label{sec:intro}

The concept of \emph{well-designed pattern trees} was introduced by Letelier et al. \cite{letelier2013static} as a convenient graphic representation of conjuctive queries extended by the optional operator. The nodes of such a tree correspond to the queries, while the tree itself represents the optional extensions. Well-designed pattern trees have been studied from a complexity point of view in several aspects. One of the most interesting problems in the context of query languages is the \emph{generation problem}, that is, generating the solutions one after the other without repetition. 

\paragraph{Previous work}
The generation problem was studied for First-Order and Conjunctive Queries \cite{bulatov2012enumerating,durand2014enumerating,kazana2013enumeration,
segoufin2013enumerating} and for well-designed pattern trees \cite{letelier2013static}. Recently, Kr\"oll et al. \cite{kroll2016complexity} initiated a systematic study of the complexity of the generation problem of well-designed pattern trees. They identified several tractable and intractable cases of the problem both from a classical and from a parameterized complexity point of view. One class of pattern trees however remained unclassified. For a class $\cC$ of conjunctive queries, a well-designed pattern tree $T$ is \emph{globally in $\cC$} if for every subtree $T'$ of $T$ the corresponding conjunctive query is also in $\cC$. The \emph{treewidth} of a conjunctive query is the treewidth of its Gaifman-graph \cite{GD86}. In \cite{kroll2016complexity}, the complexity of the generation problem for the class of well-designed pattern trees falling globally in the class of queries of treewidth at most $k$ and having $c$-semi-bounded interface was left open (see~\cite[Table 1 on page 16]{kroll2016complexity}).

At the Dagstuhl Seminar 19211 ``Enumeration in Data Management'', Kr\"oll proposed an open problem on the generation of clause sequences of CNF formulas \cite[Problem 4.7]{boros2019enumeration}. The problem is motivated by the fact that it can be reduced to the above mentioned unsolved case of pattern trees, thus any bound on the generation complexity would be helpful in understanding the general problem. A \emph{generation algorithm} outputs the objects in question one by one without repetition. We call it a \emph{polynomial delay} procedure if the computing time between any two consecutive outputs is bounded by a polynomial of the input size. We call it \emph{incrementally polynomial}, if for any $k$ the first $k$ objects can be generated in polynomial time in the input size and $k$. Finally, it is called \emph{total polynomial} if all $N$ objects are generated in polynomial time in the input size and $N$. 

The problem studied in this paper can be formalized as follows. Let $V=\{x_1,\dots,x_n\}$ be a set of $n$ Boolean variables and $\Phi=C_1\wedge\dots\wedge C_m$ be a CNF in these variables with clauses $C_1,\ldots,C_m$. For an assignment $\bm{a}:V\rightarrow\{0,1\}$, the corresponding binary sequence $\sigma_\Phi(\bm{a})=(C_1(\bm{a}),\dots,C_m(\bm{a}))$ is called a \emph{signature}\footnote{We prefer the term \emph{signature} over the term \emph{clause sequence}  proposed by Kr\"ol, since it is a binary string, not a sequence of clauses. Therefore we use the term \emph{signature} in the rest of the paper.} of $\Phi$, that is, $C_i(\bm{a})=1$ if clause $C_i$ evaluates to $1$ under assignment $a$, and $C_i(\bm{a})=0$ otherwise. In particular, this means that $\Phi$ is satisfiable if and only if there exists some assignment $\bm{a}$ with $\sigma_\Phi(\bm{a})=(1,\ldots,1)$. Moreover, MAX-SAT and MIN-SAT can be encoded by asking for the signature with the largest and smallest sum of elements, respectively. 

As an example, consider the CNF formula $\Phi = C_1\wedge C_2\wedge C_3\wedge C_4$, where $C_1 =x_1\vee \bar{x}_3$, $C_2 =\bar{x}_2$, $C_3=x_1\vee x_2\vee x_3$ and $C_4=x_2\vee \bar{x}_3$. 
Then assignment
$\bm{a_1}=\{x_1 \mapsto 1, x_2  \mapsto 1,x_3 \mapsto 1\}$ 
leads to signature 
$\sigma_\Phi(\bm{a_1})=(1,0,1,1)$, 
while assignment 
$\bm{a_2}=\{x_1 \mapsto 0, x_2  \mapsto 0,x_3 \mapsto 1\}$ 
leads to signature 
$\sigma_\Phi(\bm{a_2})=(0,1,1,0)$.
It is easy to see that $\Phi$ has six different signatures. In general, if the number of signatures is $\Omega(2^n)$, then generating them in total polynomial time is not difficult. However, their number may be $o(2^n)$, presenting a potential challenge for generation.

Given a CNF $\Phi=C_1\wedge\dots\wedge C_m$, we denote by $\dim(\Phi)=\max_{i=1,\dots,m} |C_i|$, and call $\Phi$ 
a $d$-CNF if $\dim(\Phi)\leq d$. The \emph{number of clauses} and the \emph{number of literals} appearing in $\Phi$ are denoted by $|\Phi|$ and $\|\Phi\|$, respectively. Vectors are written using bold fonts throughout, e.g. $\bm{x}$. The problem asked in \cite{boros2019enumeration} is for $d$-CNF formulas where $d$ is a fixed positive integer, but we also consider the same problem for general CNFs.

%
%
%
%
\begin{center}
\begin{minipage}[t]{0.48\textwidth}
\begin{tcolorbox}[boxsep=2pt,left=4pt,right=4pt,top=4pt,bottom=4pt]
\noindent\underline{\textsc{Generation of signatures} ($GS(\Phi)$)}

\vspace{0.2cm}

\noindent\textbf{Input:} A CNF $\Phi$.

\noindent\textbf{Output:} All possible signatures of $\Phi$.
\end{tcolorbox}
\end{minipage}
\end{center}

Motivated by MAX-SAT and MIN-SAT, we also consider maximal and minimal signatures. A signature of a CNF $\Phi$ is called \emph{maximal} (resp. \emph{minimal}) if an inclusionwise maximal (resp. minimal) subset of the clauses takes value 1.

\bigskip

\noindent \begin{minipage}[t]{0.48\textwidth}
\begin{tcolorbox}[boxsep=2pt,left=4pt,right=4pt,top=4pt,bottom=4pt]
\noindent\underline{\textsc{Generation of maximal signatures}}

\vspace{0.2cm}

\noindent\textbf{Input:} A CNF $\Phi$.

\noindent\textbf{Output:} All possible maximal signatures of $\Phi$.
\end{tcolorbox}
\end{minipage}
\hfill
\begin{minipage}[t]{0.48\textwidth}
\begin{tcolorbox}[boxsep=2pt,left=4pt,right=4pt,top=4pt,bottom=4pt]
\noindent\underline{\textsc{Generation of minimal signatures}}

\vspace{0.2cm}

\noindent\textbf{Input:} A CNF $\Phi$.

\noindent\textbf{Output:} All possible minimal signatures of $\Phi$.
\end{tcolorbox}
\end{minipage}

\paragraph{Our results}

We show that $GS(\Phi)$ can be solved in incremental polynomial time for formulas with a bounded dimension, thus answering the open problem posed by Kr\"oll, and with polynomial delay for the class of tractable CNF formulas. For the class of formulas with bounded dimension and co-occurrence, we derive a faster incremental polynomial algorithm. We also show that generating maximal signatures is NP-hard, while minimal signatures can be generated with polynomial delay.

\paragraph{Organization} 

Our algorithm with polynomial delay for the class of tractable CNF formulas is given in Section~\ref{sec:tractable}. Section~\ref{sec:bd} discusses CNFs with bounded dimension: an incremental polynomial algorithm is presented in Section~\ref{sec:bdco} for CNFs with bounded dimension and co-occurrence, while our main result answering the question of Kr\"oll is presented in Section~\ref{sec:bdb}. The generation of maximal and minimal clause sequences is considered in Section~\ref{sec:minmax}. Finally, we conclude the paper in Section~\ref{sec:conc}, where a `reversed' variant of the problem is proposed as an open question.


%
%
%
%

%
%
%
%
%

%
%

\section{Tractable CNFs}
\label{sec:tractable}

Given a CNF $\Phi=\bigwedge_{C\in\cC}C$, a CNF $\Psi=\bigwedge_{C\in\cC'}C$ is called a \emph{sub-CNF} of $\Phi$ if $\cC'\subseteq\cC$, and denoted by $\Psi\subseteq \Phi$.  We call a family of CNFs \emph{tractable} if for any CNF $\Phi$ in this family the satisfiability of any sub-CNF of $\Phi$ can be decided in polynomial time even after fixing any subset of the variables at arbitrary values. For example, the classes of 2-CNFs or Horn CNFs are tractable. 

\begin{theorem}\label{t-tractable}
If $\Phi$ belongs to a tractable family and has $m$ clauses, then its signatures can be generated with a delay of $O(m)$ SAT-calls. 
\end{theorem}
\begin{proof}
The idea is to apply the so-called `flashlight' approach in the signature space, using SAT as a `flashlight' \cite{boros2004algorithms}. Let $\Phi=\bigwedge_{i=1}^m C_i$. We are going to build a binary tree in which the paths from the root to the vertices of the tree correspond to binary values of initial segments of the set of clauses, that is, $C_1,\dots,C_k$ for some $1\leq k\leq m$. There exists a signature with this prefix if and only if the CNF formed by the clauses set to value one in this sequence is satisfiable even after all the forced fixing of variables that appear in clauses whose value is zero (note that a clause has value $0$ if and only if all the literals in it are $0$). If such a CNF is not satisfiable, we backtrack and do not explore the subtree rooted at this vertex as there exists no signature with this prefix. If the CNF is satisfiable, we continue building the corresponding subtree which in this is guaranteed to contain at least one signature. The algorithm will not backtrack above this vertex before outputing all (at least one) signatures in this subtree. It is not difficult to verify that after at most $2m$ calls to SAT we can output a new signature not generated before. After outputting the last signature, the procedure terminates after at most $m$ SAT calls.
\end{proof}

\begin{remark} \label{rem:union}
Let us remark that the family of monotone CNFs is tractable, but for this case there is a more efficient polynomial delay generation of the signatures. Indeed, in this case we can view a clause as a subset of the variables. Consequently, the set of zeros in a signature corresponds to a union of clauses. We claim that all such unions can be generated with $O(nm)$ delay, where $m=|\Phi|$ is the number of clauses, implying that all signatures of $\Phi$ can be generated with polynomial delay.

To see this claim, we represent unions as leaves of a binary tree of depth $n$ (nodes correspond to variables), where we construct only the vertices that are on paths to the leaves. Besides the binary tree, we keep the leaves in a last-in-first-out queue\footnote{The size of the queue can be exponential in $n$ as it contains the leaves of the binary tree that is being built.}. Initially, leaves correspond to individual clauses of $\Phi$. Each time before outputting the first union $U$ in the queue, we check for all clauses $C\in\Phi$ if $C\cup U$ is a new union or not by using our binary tree. This takes $O(n)$ time for one clause, and $O(nm)$ time for all the clauses of $\Phi$. Whenever a new union is found, it is added to the tree and the queue as a last element. After this, we output $U$ and remove it from the queue. It is not difficult to verify that this gives us an $O(nm)$ delay generation of all unions. Note that in this case Theorem \ref{t-tractable} guarantees only an $O(\| \Phi\|m)$ delay, because every SAT call requires $O(\| \Phi\|)$ time.
\end{remark}

\section{CNFs with bounded dimension}
\label{sec:bd}

\subsection{Bounded co-occurrence}
\label{sec:bdco}

Given a CNF $\Phi$, we denote by $H_\Phi=(\Phi,E)$ the \emph{conflict graph} of $\Phi$. The vertices of $H_\Phi$ are the clauses of $\Phi$ and edges are exactly the \emph{conflicting} pairs of clauses, i.e., pairs $(C_i,C_j)$ for which there exists a literal $u\in C_i$ such that $\bar{u}\in C_j$. 

Let $S\subseteq \Phi$ be a maximal independent set of $H_\Phi$, and let $L(S)=\bigcup_{C_i\in S} C_i$ denote the set of literals appearing in the clauses of $S$. We define a partial assignment $\bm{a_S}:L(S)\rightarrow\{0,1\}$ by setting all literals of $L(S)$ to zero (and hence the complementary literals are set to $1$). The \emph{signature associated to $S$} is then defined as $\sigma_\Phi(S):=\sigma_\Phi(\bm{a_S})=(y_1,\dots,y_m)\in\{0,1\}^m$. The coordinates of $\sigma_\Phi(S)$ are well-defined as $y_i=0$ if and only if $C_i\in S$ for $i=1,...,m$. We will dismiss the subscript $\Phi$ whenever the CNF in question is clear from the context. Note that for different maximal independent sets $S\neq S'$ of $H_\Phi$ we have $\sigma(S)\neq \sigma(S')$. 
It is worth mentioning that all maximal independent sets of $H_\Phi$ can be generated with polynomial delay \cite{tsukiyama1977new,makino2004new}, which is hence a good start for CNF signature generation. 

Assume that $\Phi$ has bounded dimension, i.e., for a constant $d$ we have $|C_i|\leq d$ for all $i=1,...,m$. Let us define $\cX_j=\{C_i\in \Phi\mid x_j\in C_i \text{ or } \bar{x}_j\in C_i\}$. We say that $\Phi$ is of \emph{$\omega$-bounded co-occurrence} if $|\cX_j|\leq \omega$ for $j=1,...,n$ and $\omega$ is a fixed constant. 

\begin{theorem} \label{thm:bdco}
If $\Phi$ has bounded dimension and co-occurrence, then its signatures can be generated in incremental polynomial time. 
\end{theorem}
\begin{proof}
Let us construct greedily a maximal induced matching $M\subseteq E$ in $H_\Phi$. Note that $H_\Phi$ has at least $2^{|M|}$ maximal independent sets (and hence at least this many signatures can be generated with polynomial delay, as explained above). We denote by $W\subseteq \Phi$ the set of clauses that have edges in $H_\Phi$ connecting them to some of the clauses covered by $M$, and set $U=\Phi\setminus W$. Note that $U$ is an independent set in $H_\Phi$. 

Assume that $\mu=|M|$, $|C_i|\leq d$ for all $i=1,...,m$, and $|\cX_j|\leq \omega$ for all $j=1,...,n$. According to our assumptions, $d$ and $\omega$ are fixed constants.  Observe that with these notations we have $|W| \leq 2\mu d \omega$. We denote by $n'$ the number of variables involved in clauses of $W$. Note that we have $n'\leq d|W|$.



We denote by $L'$ the (possibly empty) set of variables that are monotone in $\Phi$ and appear only in clauses of $U$ (some variables appear only positively while some others appear only negatively). Let us first set all literals in $L'$ to $0$, and consider the resulting CNF $\Phi'$  in $n'$ variables. We generate with polynomial delay the maximal independent sets $S_\ell$, $\ell=1,...,k$ of $H_{\Phi'}$, and the corresponding signatures $\sigma(S_\ell)$, $\ell=1,...,k$. Now we have $k\geq 2^\mu \geq ( 2^{n'})^{1/2 d^2\omega}$, and thus we can try all binary assignments to the $n'$ variables in $O(mnk^{2d^2\omega})$ time, and see if we get some more signatures. 

Assume we get $k'\geq k$ distinct signatures. By switching the literals in $L'$, we may get new signatures, resulting from changing some of the zeros in a signature to one. For any partial assignment to the $n'$ variables, this is a set-union generation problem that can be solved with polynomial delay, see Remark~\ref{rem:union}. We may get in this way the same signature multiple times, but no more than $k'$ times, and thus at this stage the additional signatures are also generated in incremental polynomial time. 
\end{proof}

\subsection{Unbounded co-occurrence}
\label{sec:bdb}

In the previous section, we considered CNFs with bounded dimension and co-occurrence. The running time of the algorithm provided by Theorem~\ref{thm:bdco} depends exponentially on $\omega$, hence it is not suitable for handling the general case. In the present section, a more general procedure is given based on a different approach. 

For a CNF $\Phi$, we denote by $G_\Phi=(\Phi,E)$ the so called \emph{dual graph} of $\Phi$ \cite{SAMER201050}. The vertices of $G_\Phi$ are the clauses of $\Phi$ and edges are exactly the pairs of clauses $(C_i,C_j)$ for which there exists a variable that occurs in both $C_i$ and $C_j$ (complemented or not). If $S\subseteq \Phi$ is an independent set of $G_\Phi$, then the clauses of $S$ have pairwise disjoint sets of variables involved.

\begin{theorem}\label{t-boundeddim}
There exists an algorithm $\mathfrak{A}$ that generates the signatures of a CNF $\Phi$ consisting of $m$ clauses in $n$ binary variables in $O(dm^2nk^{\binom{d}{2}})$ total time, where $d=\dim(\Phi)$ and $k$ is the number of signatures. 
\end{theorem}
\begin{proof}
We prove the claim by induction on $d$. For $d\leq 2$ the claim follows by Theorem \ref{t-tractable}. 

Assume now that we already proved the claim for all $d'<d$, and let us consider a CNF $\Phi=C_1\wedge C_2\wedge \dots \wedge C_m$ with $\dim(\Phi)=d$. Let us associate to $\Phi$ its dual graph $G_\Phi$ as defined above. Let $S\subseteq V(G_\Phi)$ be a maximal independent set of $G_\Phi$. Such a set can be obtained by a simple greedy procedure in polynomial time in the size of $\Phi$. Note that clauses in $S$ involve pairwise disjoint sets of variables, due to the fact that $S$ is an independent set of $G_\Phi$. Thus, we can choose a literal $u_C\in C$ for each clause $C\in S$, set all other literals in $C$ to zero, set all other variables not occurring in clauses of $S$ to zero, and make all possible truth assignment to the literals $u_C$, $C\in S$. This way we obtain $k_0=2^{|S|}$ different binary signatures of $\Phi$. Note that we can output these $k_0$ signatures with polynomial delay. 

The total number of variables involved in clauses of $S$ is $n'\leq d|S|$. Hence we can assign in all possible ways values to these variables, and produce $2^{n'}$ subproblems $\Phi_j$, $j=1,...,2^{n'}$ in the remaining variables in $O(mn 2^{n'})=O(mn k_0^d)$ time which is polynomial in the input size and $k_0$, since $d$ is a fixed constant. Each of these residual problems is of dimension at most $d-1$. Indeed, each of the clauses not in $S$ shares at least one variable with the clauses of $S$, since $S$ is a maximal independent set of $G_\Phi$, and now that shared variable is fixed at a binary value. 

We apply algorithm $\mathfrak{A}$ to each of the residual sub-CNFs $\Phi_j$, $j=1,...,2^{n'}$, one by one. This way we produce signatures that extend the pattern on $S$ defined by $x^j\in\{0,1\}^{n'}$, for all $j=1,...,2^{n'}$ one by one. We may produce the same signature in this way again and again, but no more than $2^{n'}$ times. Since $2^{n'}=O(k_0^d)$,  we can show that this procedure works in total polynomial time. 

To see this let us introduce some additional notation. We denote by $X_j\subseteq Y=\{0,1\}^{n'}$, $j=1,...,2^{|S|}$ the nonempty sets of (partial) assignments that produce the same signature on the clauses of $S$. For $\bm{x}\in Y$, let us denote by $\Phi(\bm{x})$ the residual CNF, and by $k(\bm{x})$ the number of signatures of $\Phi(\bm{x})$. We denote by $g(\Psi)$ the running time of the above described recursive algorithm on CNF $\Psi$ and let
$G(m,n,d,k)$ be the maxima of $g(\Psi)$ over all CNFs with at most $m$ clauses on $n$ variables having $\dim(\Psi)\leq d$ and having at most $k$ signatures. 

The total computational time in the first phase of the above procedure that ends with producing a list of $2^{n'}$ residual CNFs, each of $\dim \leq d-1$ is bounded by 
\[
O(m^2n)+O(mnk_0)+O(mnk_0^d) \leq Km^2nk_0^d
\]
for a suitable constant $K$ that does not depend on $m$, $n$, and $k_0$. 
The first term on the left hand side is the time to build $G_{\Phi}$ and to find a maximal independent set $S$. The second term is the time we need to generate the $k_0$ initial signatures. The third term is the time to generate the $2^{n'}\leq k_0^d$ subproblems. 

For $\bm{x}\in X_j$ and $\bm{x'}\in X_{j'}$ with $j\neq j'$ the CNFs $\Phi(\bm{x})$ and $\Phi(\bm{x'})$ cannot share signatures, since those must already differ on $S$ by the definition of the sets $X_j$ for $j=1,...,k_0$. However, for $\bm{x},\bm{x'}\in X_j$ CNFs $\Phi(\bm{x})$ and $\Phi(\bm{x'})$ may share (many) signatures. Discounting the one signature we already produced with a given trace on $S$, we can still expect 
\[
k_j ~\geq~\max_{\bm{x}\in X_j} k(\bm{x})-1
\]
different signatures produced by algorithm $\mathfrak{A}$ when we use it for CNFs $\Phi(\bm{x})$, $\bm{x}\in X_j$. Thus, in total we get
\[
k ~=~ k_0 +k_1 + \cdots + k_{2^{|S|}}
\]
different signatures for $\Phi$. The total running time on CNFs $\Phi(\bm{x})$, $\bm{x}\in X_j$ can be bounded by 
\[
\sum_{\bm{x}\in X_j} g(\Phi(\bm{x})) ~\leq~ |X_j| G(m,n,d-1,k_j).
\]
Thus, for the total running time of algorithm $\mathfrak{A}$ on $\Phi$ we get
\begin{align*}
g(\Phi) \leq G(m,n,d,k) &\leq~ Km^2nk_0^d + \sum_{j=1}^{k_0} |X_j|G(m,n,d-1,k_j)\\
 &\leq ~ Km^2nk_0^d + k_0^dG(m,n,d-1,k),
\end{align*}
where for the last inequality we used $k_j\leq k$ for all $j=1,...,k_0$, implying $G(m,n,d-1,k_j)\leq G(m,n,d-1,k)$, which allows this quantity to be factored out of the sum, that can be then upper bounded by $\sum_{j=1}^{k_0}|X_j|= 2^{n'}\leq k_0^d$. Using this we can show by induction on $d$ that 
\[G(m,n,d,k) \leq Ldm^2nk^{\binom{d}{2}} \] 
for some constant $L$ (we will choose $L\geq K$) which will complete the proof of our claim. Now
\begin{align*}
G(m,n,d,k) &\leq ~ Km^2nk_0^d + k_0^dG(m,n,d-1,k)\\
&\leq ~ Km^2nk_0^d + k_0^d L(d-1)m^2nk^{\binom{d-1}{2}}\\
&\leq ~ Lm^2nk^d + k^d L(d-1)m^2nk^{\binom{d-1}{2}}\\
&\leq ~ Lm^2nk^d + L(d-1)m^2nk^{\binom{d-1}{2}+d} \leq Ldm^2nk^{\binom{d}{2}}.
\end{align*}
\end{proof}

\begin{corollary}
The algorithm $\mathfrak{A}$ constructed in the above proof in fact works in incremental polynomial time.
\end{corollary}

\begin{proof}
Using the above theorem, we can prove this claim by induction on the dimension $d$. When $d=1$, the claim is trivially true. 

Consider now the general case, as in the proof of the above theorem. As we remarked there, producing the first $k_0=2^{|S|}$ signatures in fact can be done with polynomial delay. After this we start processing the CNFs $\Phi(\bm{x})$ for $\bm{x}\in X_j$, $j=1,...,k_0$. Note that the signatures produced from $\Phi(\bm{x})$, $\bm{x}\in X_j$ and $\Phi(\bm{x'})$, $\bm{x'}\in X_{j'}$ are all different if $j\neq j'$. Note also that $\dim(\Phi(\bm{x}))\leq d-1$ for all $\bm{x}\in X_j$, $j=1,...,k_0$, and thus we can assume by induction that their signatures can be produced in incremental polynomial time in the size of $\Phi(\bm{x})$, which is bounded by the size of $\Phi$. Thus, if $X_j=\{\bm{x_1},...,\bm{x_\ell}\}$, then we can produce $k(\bm{x_1})$ new signatures in incremental polynomial time, in fact regardless how many we produced previously (including the $k_0$ we have from the first phase.) Let us denote by $q(m,n,k(\bm{x_1}))$ the polynomial bounding the total time processing $\Phi(\bm{x_1})$. If $k(\bm{x_2})>k(\bm{x_1})$, then maybe the first $k(\bm{x_1})$ signatures produced from $\Phi(\bm{x_2})$ coincide with the ones we already generated from $\Phi(\bm{x_1})$, but still after at most $q(m,n,k(\bm{x_1}))$ time we get a new signature. In the worst case, we have $k_j=k(\bm{x_1})\geq k(\bm{x_i})$ for all $\bm{x_i}\in X_j$, $i\neq 1$, in which case processing $\Phi(\bm{x_i})$, $i=2,...,\ell$ may not produce any new signatures. Since $\ell\leq k_0^d$, this means that the largest gap between the output of the last signature of $\Phi(\bm{x_1})$ and next new signature is not more than $k_0^d q(m,n,k(\bm{x_1}))$, at a moment when we have already produced $k'\geq k_0+k(\bm{x_1})$ signatures. Thus this largest time gap between two outputs is still bounded by a polynomial of the input size $O(mn)$ and the number of signatures $k'\geq k_0+k(\bm{\bm{x_1}})$ produced so far.
\end{proof}

\section{Generating maximal and minimal signatures}
\label{sec:minmax}

Generation of maximal signatures is difficult as it includes SAT as a special case.

\begin{theorem}
Generating all maximal signatures is NP-hard.
\end{theorem}
\begin{proof}
Let us consider a CNF $\Phi$, and observe that its unique maximal signature is the all-one vector if and only if $\Phi$ is satisfiable. Hence any total polynomial time algorithm generating the maximal signatures would detect satisfiability of $\Phi$. As SAT is difficult in general \cite{cook1971complexity}, the theorem follows.
\end{proof}

It turns out that minimal signatures can be generated efficiently.

\begin{theorem}
Minimal signatures can be generated with polynomial delay.
\end{theorem}
\begin{proof}
We claim that there is a one-to-one correspondance between minimal signatures of a CNF $\Phi$ and maximal independent sets of its conflict graph $H_\Phi$. Since $H_\Phi$ can be built in polynomial time from $\Phi$ and maximal independent sets of a graph can be generated with polynomial delay \cite{tsukiyama1977new,makino2004new}, this would prove the theorem.

To see the above claim, assume first that a signature $\sigma=\{\sigma_C\mid C\in\Phi\}$ is a minimal signature of $\Phi$. Note that the set $S=\{C\in\Phi\mid \sigma_C=0\}$ is an independent set in $H_\Phi$. For any $C\in\Phi$ with $\sigma_C=1$ there must exist a conflict between $C$ and some $C'\in S$, since otherwise we could set $\sigma_C$ to zero without forcing any of the clauses in $S$ to change their values, contradicting the minimality of $\sigma$. Thus $S$ must be a maximal independent set.

The other direction follows from the fact that if $S$ is a maximal independent set of $H_\Phi$ and we set all the clauses in $S$ to zero, then all other clauses of $\Phi$ are forced to take value one due to the conflicts between $S$ and other vertices of $H_\Phi$.
\end{proof}

\section{Conclusions}
\label{sec:conc}

In this paper we show that all signatures of a given CNF with a bounded dimension can be generated in incremental polynomial time, answering an open problem posed by Kr\"oll ~\cite[Problem 4.7]{boros2019enumeration}. A faster incremental polynomial algorithm is provided for the class of formulas where both the dimension and the co-occurrence are bounded. Moreover, it is also shown that the same task can be done with polynomial delay if the input CNF is from a tractable class (in this case no bound on dimension or co-occurrence is necessary).  Finally, it is proved that generating maximal signatures is NP-hard, while minimal signatures can be generated with polynomial delay.

In this context it is interesting to note that given a 3-CNF $\Phi$ with $m$ clauses and the vector $\bm{y}=(1,1,...,1)\in \{0,1\}^m$ it is NP-hard to test whether $\bm{y}$ is a signature of $\Phi$, or not ($\bm{y}$ is a signature if only if $\Phi$ is satisfiable). On the other hand, our results show that generating all signatures of $\Phi$ can be done in incremental polynomial time. This is a rather unusual behavior for a generation problem. Typically, if all solutions of a given problem can be generated in incremental polynomial time, checking if a given candidate is a solution or not is computationally easy.

An additional problem connected to CNF signatures was stated at the Dagstuhl Seminar 19211 by Gy. Tur\'an. Given a set $S \subseteq  \{0,1\}^m$, does there exist a CNF with $m$ clauses such that $S$ is exactly its set of all signatures? If yes, can such a CNF be computed efficiently? This `reverse' problem (get the signatures, output clauses) to the problem presented in this paper (get the clauses, output signatures) is to the best of our knowledge completely open.

\medskip
\paragraph{Acknowledgements} Krist\'of B\'erczi was supported by the J\'anos Bolyai Research Fellowship of the Hungarian Academy of Sciences and by the ÚNKP-19-4 New National Excellence Program of the Ministry for Innovation and Technology. Ond\v{r}ej \v{C}epek and Petr Ku\v{c}era gratefully acknowledge a support by the Czech Science Foundation (Grant 19-19463S). Projects no. NKFI-128673 and no. ED\_18-1-2019-0030 (Application-specific highly reliable IT solutions) have been implemented with the support provided from the National Research, Development and Innovation Fund of Hungary, financed under the FK\_18 and the Thematic Excellence Programme funding schemes, respectively. This work was supported by the Research Institute for Mathematical  Sciences, an International Joint Usage/Research Center located in Kyoto University. 

\bibliographystyle{abbrv}
\bibliography{enumeration}

\end{document}